\documentclass[journal]{IEEEtran}

\usepackage{setspace}

\usepackage{graphics,
           psfrag,
           epsfig,
           amsthm,
           cite,
           amssymb,
           url,
           dsfont,
           subfigure,
           algorithm,
           algorithmic,
           balance,
           enumerate,
           color,
           setspace
}
\usepackage{amsmath}%[centertags][tbtags][sumlimits][nosumlimits][intlimits][namelimits][nonamelimits]

\usepackage{epstopdf}

\newtheorem{lemma}{Lemma}

\newtheorem{theorem}{Theorem}

\newtheorem{remark}{Remark}

\newcommand{\eref}[1]{(\ref{#1})}
\newcommand{\sref}[1]{Section~\ref{#1}}
\newcommand{\appref}[1]{Appendix~\ref{#1}}
\newcommand{\fref}[1]{Figure~\ref{#1}}

\newcommand{\cref}[1]{Constraint~\ref{#1}}
\newcommand{\thref}[1]{Theorem~\ref{#1}}
\newcommand{\lref}[1]{Lemma~\ref{#1}}

\newcommand{\algref}[1]{Algorithm~\ref{#1}}

\newcommand{\ignore}[1]{}

\IEEEoverridecommandlockouts
\usepackage{lettrine}

\begin{document}

\title{\vspace{-.5cm}Hybrid Radio/Free-Space Optical Design for Next Generation Backhaul Systems}

\author{
\authorblockN{Ahmed Douik, \textit{Student Member, IEEE}, Hayssam Dahrouj, \textit{Senior Member, IEEE},\\ Tareq Y. Al-Naffouri, \textit{Member, IEEE}, and Mohamed-Slim Alouini, \textit{Fellow, IEEE}\vspace{-.8cm}}

\thanks {A part of this paper is published in proc. of IEEE International Conference on Communication (ICC' 15) Workshop Next Generation Backhaul/Front haul Networks (BackNets' 15), London, UK, June 2015.

Ahmed Douik is with the Department of Electrical Engineering, California Institute of Technology, Pasadena, CA 91125 USA (e-mail: ahmed.douik@caltech.edu).

Hayssam Dahrouj is with the Department of Electrical Engineering, Effat University, Jeddah 22332, Saudi Arabia (e-mail: hayssam.dahrouj@kaust.edu.sa).

T. Y. Al-Naffouri is with King Abdullah University of Science and Technology, Thuwal 23955-6900, Saudi Arabia, and also with King Fahd University of Petroleum and Minerals, Dhahran 31261, Saudi Arabia (e-mail: tareq.alnaffouri@kaust.edu.sa).

M.-S. Alouini is with the Division of Computer, Electrical and Mathematical Sciences and Engineering, King Abdullah University of Science and Technology, Thuwal 23955-6900, Saudi Arabia (e-mail: slim.alouini@kaust.edu.sa).
}
}

\maketitle

\IEEEoverridecommandlockouts

\begin{abstract}
The deluge of date rate in today's networks imposes a cost burden on the backhaul network design. Developing cost efficient backhaul solutions becomes an exciting, yet challenging, problem. Traditional technologies for backhaul networks include either radio-frequency backhauls (RF) or optical fibers (OF). While RF is a cost-effective solution as compared to OF, it supports lower data rate requirements. Another promising backhaul solution is the free-space optics (FSO) as it offers both a high data rate and a relatively low cost. FSO, however, is sensitive to nature conditions, e.g., rain, fog, line-of-sight. This paper combines both RF and FSO advantages and proposes a hybrid RF/FSO backhaul solution. It considers the problem of minimizing the cost of the backhaul network by choosing either OF or hybrid RF/FSO backhaul links between the base-stations (BS) so as to satisfy data rate, connectivity, and reliability constraints. It shows that under a specified realistic assumption about the cost of OF and hybrid RF/FSO links, the problem is equivalent to a maximum weight clique problem, which can be solved with moderate complexity. Simulation results show that the proposed solution shows a close-to-optimal performance, especially for practical prices of the hybrid RF/FSO links.
\end{abstract}


\section{Introduction}\label{sec:int}

\lettrine[lines=2]{C}{ellular} networks, flooded by an enormous demand for mobile data services, are expected to undergo a fundamental transformation. In order to significantly increase the data capacity, coverage performance, and energy efficiency, next generation mobile networks ($5$G) \cite{546843} are expected to move from the traditional single, high-powered base-station (BS) to the deployments of multiple overlaying access points of diverse sizes, i.e., microcell, picocell, femtocell, etc., using different radio access technologies. To efficiently manage the resulting high levels of interference, connecting BSs through efficient backhauling becomes a critical component in the network planning. The big increase in small cell deployment, however, necessitates a considerable amount of backhaul communications in order to share the data streams between all BSs across the network \cite{6736746}. Giving that the links are capacity limited, upgrading the backhaul and increasing its ability to support the tremendous amount of data is a necessity \cite{13050958}. Therefore, choosing the suitable technology(ies) and design of the backhaul network is of great interest, especially that its deployment cost is a dominant cost driver for many operators, e.g., approximately between $30\%$ and $50\%$ of the total operating costs for $4$G systems  \cite{591842}. With the deployment of multiple small cells expected in $5$G, the implementation costs are believed to get even higher \cite{13050958}. This paper proposes a cost efficient backhaul solution for next generation backhaul systems using techniques from graph theory.

\subsection{Backhaul Technologies}

Traditional technologies for the backhaul network design include copper, microwave radio links (RF), and optical fibers (OF). The leased T$1$/E$1$ copper lines is the most widely used backhaul technology with approximately $90\%$ of the total backhaul deployment in the US \cite{5473878}. With a provided data rate of $1.544$ Mbit/s for T$1$ and $2.048$ Mbit/s for E$1$ \cite{6777766}, copper lines provide satisfying data rates for voice traffic for $2$G networks. However, to achieve the data rate demand of $3$G traffic and beyond, multiple parallel connections are required which results in a price growing linearly with the provided capacity. For high data rates, copper lines become expensive and hence not a suitable solution for the backhaul upgrade of next generation systems, i.e., $5$G.

Microwave radio is the second most used technology for the backhaul network design as it represents $6\%$ of the total used transport media \cite{5473878} in the US. The RF technology represents a reasonable alternative to copper, especially in locations in which the deployment of wired connections is challenging. However, such solution requires an initial investment in the licensed part of the spectrum \cite{5185525}. Moreover, low frequency radio (radio waves below $6$ GHz) is limited in terms of data rates due to interference problems and high frequency radio (microwave \& millimetre wave (mmwave) from $6$ to $300$ GHz) are limited in the transmission coverage area.

Optical fiber (OF) backhaul links provide the highest rates over long distances, e.g., $155.52$ Mbit/s for STM-$1$, $622$ Mbit/s for STM-$4$, $2.4$ Gbit/s STM-$16$, and $9.9$ Gbit/s for STM-$64$ \cite{5473878}. However, as they are expensive to be deployed and require a considerable initial investment \cite{6226966}, they represent $4\%$ of the total backhaul deployment in the US \cite{5473878}. They, further, suffer from the drawbacks of wired connections, i.e., deployment is not always feasible, which restricts their utilization in particular applications.

Recently, the free-space optics technology (FSO) emerges as a substitute \cite{5771213} for next generation backhauls. An FSO link refers to a laser beam between a pair of photo-detector transceivers using the free-space as medium of transportation. Giving that its wavelength is in the micrometer range, which is an unlicensed band, FSO links are not only free to use but also immune to electromagnetic interference generated by the RF links. The high bandwidth and interference immunity features make an FSO link up to $25$ fold more efficient than an RF link in terms of capacity \cite{1495057}. FSO particularly represents a cost-efficient solution compared to OF.

In contrast with the omni-reliability of the OF and RF links, FSO links are sensitive to weather conditions, e.g., fog, snow, and rain \cite{65841515}. Therefore, reliability becomes an important factor to address for the design of FSO-based backhaul networks. To cope with such varying reliability, and combine the advantages of RF (reliability) and FSO (capacity), the hybrid RF/FSO backhaul becomes an attractive cost-effective and reliable solution. Hybrid RF/FSO transmits, when possible, simultaneously on both the RF and FSO links. In harsh weather conditions that affect the FSO link, the data is sent solely on the RF link \cite{6876609}. Moreover, hybrid RF/FSO transceivers can be quickly deployed over several kilometers \cite{16546169} and can also be easily combined with OF links \cite{5545666}. For all previously mentioned benefits, hybrid RF/FSO is a graceful complementary option for upgrading the existing backhaul network \cite{254515158}, as further shown in our paper.

\subsection{Related Work}

In the past few years, hybrid RF/FSO attracted a significant amount of research. Most of the current work \cite{6844864,168714992012} focus on the determination of the factors affecting the FSO link performance, e.g, weather conditions, scintillation, ect., and finding solutions to improve the quality, e.g., use of multiple lasers and multiple apertures, etc. However, fundamental problems of hybrid RF/FSO architecture optimization for the backhaul network topology design are only at their beginning.

The authors in \cite{1495122} design an efficient and scalable algorithm to optimize a given physical layer objective for $2$ and $3$ optical transceivers per node with a minimum number of links. The authors in \cite{Smadi:09,4609027} propose upgrading an RF network by optimally deploying the minimal number of FSO transceivers so as to achieve a given throughput. While Kashyap {\it{et al}} \cite{1495057} design a routing algorithm for hybrid RF/FSO networks that backs up the traffic to the FSO routes when the RF links could not carry it, Rak {\it{et al}} \cite{6876609} introduce a linear integer programming model to determine routing in hybrid RF/FSO network in which the FSO link availability is varying with the weather conditions. In \cite{4357553}, the authors consider a hybrid RF/FSO system in which the RF and FSO links operate at different data rates. They derive an upper bound for the capacity per node that is asymptotically achievable for random networks.

Numerous mixed integer programming model are proposed to formulate the problem of backhaul network design using hybrid RF/FSO technology. In particular, Son {\it{et al}} \cite{5462107} present an algebraic connectivity-based formulation for the design of the backbone of wireless mesh networks with FSO links and solve it using a greedy approach that iteratively inserts nodes to maximize the algebraic connectivity. The authors in \cite{4746591} propose to maximize the network throughput by installing as many FSO links as possible under the constraint that the number of FSO links in a node is bounded. In \cite{6134071}, Ahdi {\it{et al}} introduce a mixed integer programming model to find the optimal placement of FSO links in order to upgrade an existing RF backhaul network. Similarly, reference \cite{4609027} propose to improve an existing RF backhaul network with FSO links using the minimum number of FSO links to guarantee a target network throughput when RF links are non-available due to interference.

This paper suggests upgrading a pre-deployed OF backhaul network using hybrid RF/FSO links, and, hence, is related to the concept developed in \cite{6777766,6844494,2514014}. The authors in \cite{6777766} consider the upgrade of a pre-deployed OF backhaul network using FSO links and mirrors for nodes not in line-of-sight of each other. For two link-disjoint paths networks, they formulate the problem as a mixed integer programming and extend the study in \cite{6844494} to $K$ link-disjoint paths. In \cite{2514014}, the same group of authors analyze the impact of the parameter $K$ on the design. This paper extends the concept by suggesting using hybrid RF/FSO links, considering a minimum reliability constraint, and proposing a close-to-optimal explicit solution using graph theory techniques.

\subsection{Contributions}

This paper examines the problem of upgrading a pre-deployed OF backhaul network. It considers the problem of minimizing the cost of the backhaul network by choosing either OF or hybrid RF/FSO backhaul links between the base-stations (BS) so as to satisfy data rate, connectivity, and reliability constraints. Unlike our recent works which focus on connecting BSs through $K$ link-disjoint paths in order to cope with possible link failures (see \cite{Dahrouj_backnet_magazine} for the business case of the RF backhaul, and \cite{Hybrid_Douik_ICC16} for the technical details of the proposed resilient solution), a primary concern of the current paper is to guarantee network connectivity achieved by connecting each pair of nodes in the network, possibly via multiple hops. While the deployment cost of hybrid RF/FSO links depends mainly on the expense of the hybrid RF/FSO transceivers, the implementation cost of OF links depends mostly on the distance between the two end nodes. On the other hand, OF links always satisfy the data rate and reliability constraint. The performance of hybrid RF/FSO, however, degrades with the distance and the number of installed links. The paper solves the problem using graph theory techniques by introducing the corresponding planning graph. The paper's main contribution is to provide a close-to-optimal explicit solution to the problem. The paper shows that under a specified realistic assumption about the cost of OF and hybrid RF/FSO links, the problem can be reformulated as a maximum weight clique problem, that can be globally solved using efficient algorithms \cite{9874286,16513519,13265492,6607889}.

The rest of this paper is organized as follows: \sref{sec:sys} presents the considered system model and the problem formulation. In \sref{sec:prob}, the planning problem is approximated by a more tractable one. \sref{sec:prop} illustrates the proposed solution. Before concluding in \sref{sec:conc}, simulation results are presented in \sref{sec:sim}.

\section{System Model and Problem Formulation}\label{sec:sys}

\subsection{System Model and Parameters}

\begin{figure}[t]
\centering
\includegraphics[width=.6\linewidth]{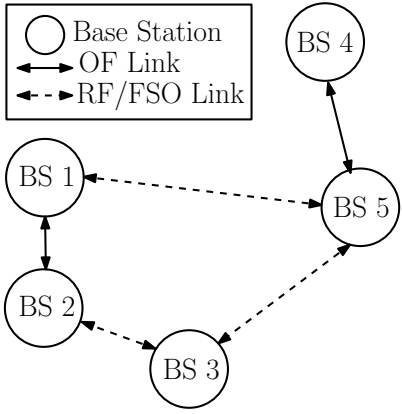}\\
\caption{Network containing $5$ base-stations connected with OF and hybrid RF/FSO links.}\label{fig:network}
\end{figure}

This paper considers a backhaul network connecting a set $\mathcal{B}=\{b_1,\ \cdots,\ b_M\}$ of $M$ base-stations with pre-deployed OF connections. All nodes (interchangeably denoting base-stations) are assumed to have a line-of-sight connections. Therefore, each node can be connected to any other node with either an OF or a hybrid RF/FSO connection, as in \fref{fig:network} which shows a network containing $5$ base-stations.

Let $P_{ij}, 1 \leq i,j \leq M$ be a binary variable indicating the existence, i.e., $P_{ij}=1$, of a pre-deployed OF link between base-stations $b_i$ and $b_j$ and $P_{ij}=0$ otherwise. For simplicity of notations, $P_{ij}$ may be also written as $P(b_i,b_j)$.

Let $\pi^{(O)}_{ij}$ and $\pi^{(h)}_{ij}$ the cost an OF and a hybrid RF/FSO link between nodes $b_i$ and $b_j$, respectively. Obviously, the functions $\pi^{(O)}_{ij}$ and $\pi^{(h)}_{ij}$ are positive and symmetric, e.g., $\pi^{(O)}_{ij}=\pi^{(O)}_{ji}$. Finally, as hybrid RF/FSO is a cost effective solution, the paper assumes that $\pi^{(h)}_{ij} \leq \pi^{(O)}_{ij} , \forall \ i,j \in \mathcal{B}$.

Let $D^{(O)}_{ij}$ and $D^{(h)}_{ij}$ be the provided data rates of an OF and a hybrid RF/FSO links between nodes $b_i$ and $b_j$, respectively. Let $D_t$ be the targeted data rate. Since the OF provides high data rates, without loss of generality, this work assumes that $D^{(O)}_{ij} \geq D_t, \forall \ i \neq j \in \mathcal{B}$.

Similarly, let $R^{(O)}_{ij}$ and $R^{(h)}_{ij}$ be the reliability of an OF and a hybrid RF/FSO links between nodes $b_i$ and $b_j$, respectively. Let $\alpha$ be the targeted reliability. As OF links are always reliable, this paper assumes that $R^{(O)}_{ij} \geq \alpha, \forall \ i \neq j \in \mathcal{B}$.

\subsection{Problem Formulation}

Let $X_{ij}, 1 \leq i,j \leq M$ be a binary variable indicating if base-stations $b_i$ and $b_j$ are connected with an OF connection. Similarly, let $Y_{ij}, 1 \leq i,j \leq M$ indicate if they are connected with a hybrid RF/FSO link. To simplify the problem formulation and constraints, this paper takes as a convention that $X_{ii}=Y_{ii}=0, \forall \ i \in \mathcal{B}$ in all the rest of the equations.

This paper considers the problems of minimizing the network deployment cost under the following constraints:
\begin{enumerate}
\item C1: Some nodes have pre-deployed OF links.
\item C2: Connections between nodes can be either OF or hybrid RF/FSO.
\item C3: Each node has a data rate that exceeds the targeted data rate.
\item C4: The reliability of each node exceeds the targeted reliability.
\item C5: Each node can communicate with any other node through single or multiple hop links.
\end{enumerate}

Let the Laplacian matrix $L$ be defined as $\mathbf{L} = \mathbf{D}-\mathbf{C}$, where $\mathbf{D}=\text{diag}(d_1,\ \cdots,\ d_M)$ is a diagonal matrix with $d_i = \sum_{j=1}^M X_{ij}+Y_{ij}$ and $c_{ij} = X_{ij}+Y_{ij}$. The diagonalization of the Laplacian matrix is given by $\mathbf{L}= \mathbf{Q}\mathbf{\Lambda} \mathbf{Q}^{-1}$, where $\mathbf{\Lambda}= \text{diag}(\lambda_1,\ \cdots,\ \lambda_M)$ with $\lambda_1 \leq \lambda_2 \leq \cdots \leq \lambda_M$. The connectivity condition C5 of the matrix can be written using the algebraic formulation proposed in \cite{25181258} as $\lambda_2 > 0$.

The following lemma introduces the cost-efficient backhaul design problem formulation:
\begin{lemma}
The problem of minimizing the cost of the backhaul network planning can be formulated as:
\begin{subequations}
\label{Original_optimization_problem}
\begin{align}
\min & \ \cfrac{1}{2} \sum_{i=1}^M \sum_{j=1}^M X_{ij}\pi^{(O)}_{ij} + Y_{ij}\pi^{(h)}_{ij} \\
{\rm s.t.\ } & X_{ij} = X_{ji} \label{eq:1} \\
&Y_{ij} = Y_{ji} \label{eq:2} \\
&X_{ij}P_{ij} = P_{ij} \label{eq:np} \\
&X_{ij}Y_{ij} = 0 \label{eq:3} \\
&\sum_{j=1}^M X_{ij}D_t + Y_{ij}D^{(h)}_{ij} \geq D_t\label{eq:5} \\
&1 - \prod_{j=1}^M (1 - X_{ij} \alpha)(1 - Y_{ij} R^{(h)}_{ij}) \geq \alpha \label{eq:nr} \\
&\lambda_2 > 0 \label{eq:6} \\
&X_{ij},Y_{ij} \in \{0,1\},\ 1 \leq i,j\leq M \label{eq:4},
\end{align}
\end{subequations}
where the optimization is over both binary variables $X_{ij}$ and $Y_{ij}$.
\label{l1}
\end{lemma}

\begin{proof}
To formulate the problem, the objective function and the system constraints C1 to C5 are expressed in terms of the variables $X_{ij}$ and $Y_{ij}$. Combining all the expressions yields the optimization problem \eref{Original_optimization_problem}. The complete proof can be found in \appref{ap1}.
\end{proof}

The optimization problem \eref{Original_optimization_problem} is equivalent to a weighted Steiner tree problem which is NP-hard with a complexity of order $2^{M^2}$. The optimal solution to such problem is referred to as the optimal planning. The rest of this paper proposes an efficient method to solve the problem \eref{Original_optimization_problem}, under the assumption that the hybrid RF/FSO connection between two nodes that are not neighbors is always more expensive than the cost of the OF connections between each node and its closest neighbour. The rationale for such assumption is that, for short distances, OF links are much cheaper than hybrid RF/FSO ones. Under this assumption, the next section shows that the solution for backhaul network design becomes mathematically tractable with a complexity of order $2^{M}$.

\section{Problem Approximation} \label{sec:prob}

As highlighted above, the original optimization problem \eref{Original_optimization_problem} is an NP-hard problem. The difficulty in solving the problem lies particularly in the structure of constraint \eref{eq:6} and in simultaneously optimizing \eref{Original_optimization_problem} over both binary variables $X_{ij}$ and $Y_{ij}$. This section presents an efficient heuristic to solve the problem under the assumption that the hybrid RF/FSO connection between two nodes that are not neighbors is always more expensive than the OF connections between each node and its closest neighbour. The assumption is motivated by the fact that, for short distances, OF links are much cheaper than hybrid RF/FSO ones. The heuristic is based on first finding the solution to the problem when only OF links can be used. Afterwards, it solves an approximate of the backhaul network planning problem via relating problem \eref{Original_optimization_problem} to solution reached by the planning problem when only OF links are allowed.

\subsection{Optimal Planning Using Optical Fiber Only}

The following lemma introduces the reduced problem when only OF links are allowed.
\begin{lemma}
The problem of backhaul design with minimum cost, when only OF links are allowed, is the following:
\begin{subequations}
\label{eq:11}
\begin{align}
\min& \ \cfrac{1}{2} \sum_{i=1}^M \sum_{j=1}^M X_{ij} \pi^{(O)}_{ij} \\
{\rm s.t.\ }& X_{ij} = X_{ji} \label{eq:9} \\
&X_{ij}P_{ij} = P_{ij} \label{eq:np2} \\
&\lambda_2 > 0 \label{eq:8} \\
&X_{ij} \in \{0,1\},\ 1 \leq i,j\leq M \label{eq:10}.
\end{align}
\end{subequations}
\label{l2}
\end{lemma}

\begin{proof}
To show this lemma, it is sufficient to show that, for backhaul network using only OF links, i.e., $Y_{ij}=0,\forall \ i,j$, constraints \eref{eq:2}, \eref{eq:3}, \eref{eq:5}, and \eref{eq:nr} become redundant. This can easily be done by noting that for a connected graph, each node is connected to, at least, another node. Since only a single OF connection is sufficient to ensure throughput and reliability, the constraints become redundant.
\end{proof}

To solve the problem mentioned above, the paper proposes to cluster BSs, according to the minimal price. First, a cluster $\mathcal{Z}$ containing all BSs is formed. For each connected nodes with pre-deployed OF links, the base-stations are merged into a single group (belonging to the big cluster $\mathcal{Z}$) and the corresponding $X_{ij}$ set to $1$. Afterwards, find the two minimum-price clusters and merge them into a single group. The cost between two clusters is defined as the minimum cost between all BS in each cluster. When two clusters are merged, the two minimum-price BSs in each cluster are connected through an OF link. The process is repeated until only one group remains in the system. In other words, the process terminates when all nodes are merged into a single cluster, i.e., $|\mathcal{Z}|=1$. The steps of the algorithm are summarized in \algref{alg1}. The following theorem characterizes the solution produced by \algref{alg1} with respect to the problem defined in \lref{l2}:

\begin{algorithm}[t]
\begin{algorithmic}
\REQUIRE $\mathcal{B}$, $P_{ij}$, and $\pi^{(O)}$.
\STATE Initialize $X_{ij}=P_{ij},\ 1 \leq i,j\leq M $.
\STATE Initialize $\mathcal{Z} = \varnothing$.
\FORALL {$b \in \mathcal{B}$}
\STATE Initialize $t=0$.
\FORALL {$Z \in \mathcal{Z}$}
\FORALL {$b^{\prime} \in Z$}
\IF{$P(b,b^{\prime}) = 1$}
\STATE $\mathcal{Z} = \mathcal{Z} \setminus \{Z\}$.
\STATE $Z = \{Z,b\}$.
\STATE $\mathcal{Z} = \{\mathcal{Z},\{Z\}\}$.
\STATE $t=1$.
\ENDIF
\ENDFOR
\ENDFOR
\IF{$t=0$}
\STATE $\mathcal{Z} = \{\mathcal{Z},\{b\}\}$.
\ENDIF
\ENDFOR
\WHILE {$|\mathcal{Z}| > 1$}
\STATE $(Z_i,Z_j) = \arg \min\limits_{\substack{Z,Z^{\prime} \in \mathcal{Z} \\ Z \neq Z^{\prime}}} \left[ \min\limits_{\substack{b \in Z \\ b^{\prime} \in Z^{\prime}}} \pi^{(O)}(b , b^{\prime}) \right] $.
\STATE $(b_i,b_j) = \arg \min\limits_{\substack{b \in Z_i \\ b^{\prime} \in Z_j}} \pi^{(O)}(b , b^{\prime})$.
\STATE $X_{ij} = X_{ji} = 1$.
\STATE $\mathcal{Z} = \mathcal{Z} \setminus \{Z_i\}$.
\STATE $\mathcal{Z} = \mathcal{Z} \setminus \{Z_j\}$.
\STATE $\mathcal{Z} = \{\mathcal{Z},\{Z_i,Z_j\}\}$.
\ENDWHILE
\end{algorithmic}
\caption{Optimal planning using only OF links}
\label{alg1}
\end{algorithm}

\begin{theorem}
The solution reached by \algref{alg1} is the optimal solution to the problem proposed in \lref{l2}. Such solution is referred to, in this paper, as the optimal OF only planning.
\label{th1}
\end{theorem}

\begin{proof}
To prove this theorem, we employ a two-stage proof. The first part of the proof shows that the solution reached by \algref{alg1} is the optimal solution to the problem proposed in \lref{l2} for a network without pre-deployed OF links. The second part of the proof extends the result for networks with pre-deployed OF connections. Therefore, we first show that \algref{alg1} produces a feasible solution to the problem. Afterwards, we show that any graph that can be reduced, using an algorithm similar to\algref{alg1}, to a single cluster includes the graph designed by \algref{alg1}. Finally, we show that any solution that cannot be reduced to a single group is not optimal. The complete proof can be found in \appref{ap3}.
\end{proof}

\subsection{Problem Approximation}

This subsection approximates the backhaul network planning problem \eref{Original_optimization_problem} under the assumption that a hybrid RF/FSO connection between two nodes that are not neighbours is more expensive than the OF links between each node and its closest neighbour. We first define $b_{i^*}$ as the closest node to base-station $b_i$ as follows:
\begin{align}
b_{i^*} = \arg \min_{\substack{b \in \mathcal{B} \\ b \neq b_i } } \pi^{(O)}(b_i,b).
\end{align}

The set of neighbours $\mathcal{N}_i$ of base-station $b_i$ is defined as the set of base-stations that are closest to base-station $b_i$, and that satisfy the connectivity condition. Mathematically, the condition can be written as:
\begin{align}
\mathcal{N}_i = \left\{b \in \mathcal{B} \setminus b_i \text{ such that } \pi^{(h)}(b_i,b) \leq \max_{b_j \in \mathcal{B}}\overline{X}_{ij}\pi^{(h)}_{ij} \right\},
\label{eq:nei}
\end{align}
where $\overline{X}_{ij},\ 1 \leq i \neq j \leq M$ is the optimal solution found in solving the OF only planning problem \eref{eq:11}.

\begin{remark}
The results presented in this paper do not depend on the definition of the set of neighbours $\overline{\mathcal{N}}_i$ of node $b_i$ as long as $\mathcal{N}_i \subset \overline{\mathcal{N}}_i$. Intuitively, as the set $\overline{\mathcal{N}}_i$ gets bigger and bigger, the approximation of the solution is more tight. For $\overline{\mathcal{N}}_i = \mathcal{B} \setminus \{b_i\}$, the proposed algorithm reduces to an exhaustive search.
\end{remark}

The assumption that two nodes that are far away from each others (i.e., not neighbours) connected with hybrid RF/FSO link generate a cost greater that the expense of the same nodes connected with OF links with their closest neighbours can be written $\forall\ (b_i,b_j) \notin \mathcal{N}_j \times \mathcal{N}_i$ as follows:
\begin{align}
\pi^{(O)}_{ii^*}+\pi^{(O)}_{jj^*} \leq \pi^{(h)}_{ij}.
\label{eq:12}
\end{align}

Let $\mathcal{R}_i = \{ b_j \in \mathcal{B} \setminus \{b_i\} \ | \ R^{(h)}_{ij} \geq \alpha\}$ be the set of nodes that satisfy, by their own, the reliability condition for node $b_i$. Based on the above assumption and definitions, the following lemma approximates the optimization problem \eref{Original_optimization_problem} under the assumption \eref{eq:12}.
\begin{lemma}
The problem of backhaul network cost minimization design using OF and hybrid RF/FSO connections can be approximated by the following problem:
\begin{subequations}
\label{Approximate_optimization_problem}
\begin{align}
\min& \ \cfrac{1}{2} \sum_{i=1}^M \sum_{j=1}^M X_{ij}\pi^{(O)}_{ij}+Y_{ij}\pi^{(h)}_{ij} \label{eq:25} \\
{\rm s.t.\ } & X_{ij} = X_{ji} \label{eq:20} \\
&Y_{ij} = Y_{ji} \label{eq:21}\\
&X_{ij}P_{ij} = P_{ij} \label{eq:np3} \\
&X_{ij}Y_{ij} = 0 \label{eq:22} \\
&\sum_{j=1}^M X_{ij}D_t + Y_{ij}D^{(h)}_{ij} \geq D_t \label{eq:23}\\
&\sum_{j=1}^M X_{ij}\tilde{\alpha} + \sum_{j \in \mathcal{R}_i}Y_{ij}\tilde{\alpha} + \sum_{j \in \overline{\mathcal{R}}_i}Y_{ij} R^{(h)}_{ij} \geq \tilde{\alpha} \label{eq:nr2} \\
&(X_{ij} + Y_{ij})\overline{X}_{ij} = \overline{X}_{ij} \label{eq:14} \\
&X_{ij},Y_{ij} \in \{0,1\},\ 1 \leq i,j\leq M,
\end{align}
\end{subequations}
where $\overline{\mathcal{R}}_i= \mathcal{B} \setminus \mathcal{R}_i$ is the complementary set of $\mathcal{R}_i$ and $\tilde{\alpha} = -\log(1-\alpha)$.
\label{l4}
\end{lemma}

\begin{proof}
To show that the original problem, it is sufficient to show that any solution to \eref{Approximate_optimization_problem} is a feasible solution to \eref{Original_optimization_problem}. Therefore, we show that constraint \eref{eq:nr2} is equivalent to constraint \eref{eq:nr} and that constraint \eref{eq:14} is included in constraint \eref{eq:6}. The complete proof can be found in \appref{ap4}.
\end{proof}

\section{Proposed Solution}\label{sec:prop}

This section proposes the solution for the approximate problem \eref{Approximate_optimization_problem}. The solution is based first on constructing the network planning graph, and then on formulating the problem \eref{Approximate_optimization_problem} as a graph theory problem that can be optimally solved with moderate complexity.

\subsection{Planning Graph}

In this section, we introduce the undirected \emph{planning} graph $\mathcal{G}(\mathcal{V},\mathcal{E})$, where $\mathcal{V}$ is the set of vertices and $\mathcal{E}$ the set of edges. Before stating the vertices construction and the edge connection, we first define the cluster $\mathcal{C}_i$ for each node $b_i, 1 \leq i \leq M$ as follows:
\begin{align}
\mathcal{C}_i = \{((X_{ij_1},Y_{ij_1}),\ \cdots,\ &(X_{ij_{|\mathcal{N}_i|}},Y_{ij_{|\mathcal{N}_i|}})), \text{ such that } \nonumber \\
\bigcup_{j \in \mathcal{N}_i}b_{j} &= \mathcal{N}_i \nonumber \\
X_{ij}P_{ij} &= P_{ij}, \forall \ j \in \mathcal{N}_i \nonumber \\
X_{ij}Y_{ij} &= 0, \forall \ j \in \mathcal{N}_i \nonumber \\
\sum_{j \in \mathcal{N}_i} X_{ij}D_t + Y_{ij}D^{(h)}_{ij} &\geq D_t \\
\sum_{j \in \mathcal{N}_i} X_{ij}\tilde{\alpha} + \sum_{j \in \mathcal{N}_i \cap \mathcal{R}_i}Y_{ij}\tilde{\alpha} &+ \sum_{j \in \mathcal{N}_i \cap \overline{\mathcal{R}}_i}Y_{ij} R^{(h)}_{ij} \geq \tilde{\alpha} \nonumber \\
(X_{ij}+Y_{ij})\overline{X}_{ij} &= \overline{X}_{ij}, \forall \ j \in \mathcal{N}_i \}.\nonumber
\end{align}

Define the weight of each element $\gamma_i \in \mathcal{C}_i$, ($\gamma_i = \{(X_{ij_1},Y_{ij_1}),\ \cdots,\ (X_{ij_{|\mathcal{N}_i|}},Y_{ij_{|\mathcal{N}_i|}})\}$), as follows:
\begin{align}
w(\gamma_i) = -\cfrac{1}{2}\sum_{j \in \mathcal{N}_i} X_{ij}\pi^{(O)}_{ij}+Y_{ij}\pi^{(h)}_{ij}.
\label{eq:15}
\end{align}

For each cluster $\gamma_i \in \mathcal{C}_i$, a vertex $v_{ij}, 1 \leq j \leq |\mathcal{C}_i|$ is generated. Two distinct vertices $v_{ij}$ and $v_{kl}$ are connected with an edge in $\mathcal{E}$ if the two following conditions are satisfied:
\begin{enumerate}
\item C1: $i \neq k$: The vertices represents different nodes in the network.
\item C2: $(X_{ik},Y_{ik}) = (X_{ki},Y_{ki})$ if $(b_i,b_k) \in (\mathcal{N}_k,\mathcal{N}_i)$: The vertices are non conflicting.
\end{enumerate}

\subsection{Proposed Algorithm}

The following theorem characterizes the solution of the approximated backhaul network planning problem \eref{Approximate_optimization_problem}.
\begin{theorem}
Let $(X^*_{ij},Y^*_{ij}), 1 \leq i,j \leq M$ be the optimal solution to the planning problem \eref{Approximate_optimization_problem} then we have $X^*_{i,j}+Y^*_{i,j}=1$ only if $(i,j) \in \mathcal{N}_j \times \mathcal{N}_i$.
\label{th2}
\end{theorem}

\begin{proof}
To show this theorem, the scenarios that can result in a violation of the desired property are identified. Using the cost optimality, connectivity constraint, and the assumption about the relative value, all such scenarios are shown to be sub-optimal to \eref{Approximate_optimization_problem}. Therefore, the optimal solution satisfies the property. The complete proof can be found in \appref{ap5}.
\end{proof}

The following theorem links the solution of problem \eref{Approximate_optimization_problem} to the planning graph.
\begin{theorem}
The solution of the approximation of the backhaul network problem \eref{Approximate_optimization_problem} using hybrid RF/FSO can be formulated as a maximum weight clique, among the cliques of size $M$ in the planning graph, in which the weight of each vertex $v_{ij}$ is the weight of the corresponding cluster $\gamma_i$ defined in \eref{eq:15}.
\label{th3}
\end{theorem}

\begin{proof}
To prove this theorem, we first show that there is a one to one mapping between the set of feasible solutions of the problem \eref{Approximate_optimization_problem} and the set of cliques of degree $M$ in the planning graph $\mathcal{G}(\mathcal{V},\mathcal{E})$. To conclude the proof, we show that the weight of the clique is equivalent to the merit function of the optimization problem \eref{Approximate_optimization_problem}. The complete proof can be found in \appref{ap6}.
\end{proof}

\subsection{Complexity Analysis}

This subsection characterizes the complexity of solving the original $0-1$ integer program proposed in \eref{Original_optimization_problem} and its relaxed version proposed in \eref{Approximate_optimization_problem}.

In order to characterize the complexity of \eref{Original_optimization_problem}, we first compute and reduce the number of variables. Initially, the number of $X_{ij}$ and $Y_{ij}$ is $M^2$ each. However, as $X_{ii}$ and $Y_{ii}$ can take arbitrary values, the number reduces to $M^2-M$ variables each. Furthermore, from constraint \eref{eq:1} and \eref{eq:2} which translate the symmetry of the problem, only half of the variables are independent. Hence, the number of free variables is $\frac{M^2-M}{2}$. Finally, the pre-deployed OF links, i.e., constraint \eref{eq:np}, limits the number of variables. In fact, the constraint $P_{ij}=1$ translates to $X_{ij}=1$ and $Y_{ij}=0$. Let $\overline{P}=\sum_{i=1}^M \sum_{j=1}^M P_{ij}$ be the number of pre-deployed links. It can clearly be seen that the number of variables of $X_{ij}$ and $Y_{ij}$ is $\frac{M^2-M-\overline{P}}{2}$. Let $\mathcal{P}=\{(i,j) \ | \ i<j, P_{ij}=0\}$ be the set of nodes that do not have a pre-deployed links. From constraint \eref{eq:3}, the variables $X_{ij}$ and $Y_{ij}$ are not independent. As only $3$ combinations are possible, they can be represented by $\frac{M^2-M-\overline{P}}{2}$ ternary variable $Z_{ij}, (i,j) \in \mathcal{P}$ defined as follows:
\begin{align}
Z_{ij} =
\begin{cases}
0 \hspace{0.5cm} &\text{if } X_{ij}=0 \text{ and } Y_{ij}=0 \\
1 \hspace{0.5cm} &\text{if } X_{ij}=1 \text{ and } Y_{ij}=0 \\
2 \hspace{0.5cm} &\text{if } X_{ij}=0 \text{ and } Y_{ij}=1.
\end{cases}
\end{align}
Therefore, the optimization problem \eref{Original_optimization_problem} can be written as follows:
\begin{align*}
\min & \ \sum_{(i,j) \in \mathcal{P}} \delta(Z_{ij}-1)\pi^{(O)}_{ij} + \delta(Z_{ij}-2)\pi^{(h)}_{ij} \nonumber \\
{\rm s.t.\ } & \sum_{j\in \mathcal{P}} \delta(Z_{ij}-1)D_t + \delta(Z_{ij}-2)D^{(h)}_{ij} \geq D_t \nonumber \\
&1 - \prod_{j\in \mathcal{P}} (1 - \delta(Z_{ij}-1)\alpha)(1 - \delta(Z_{ij}-2) R^{(h)}_{ij}) \geq \alpha \nonumber \\
\end{align*}
\begin{align}
&\lambda_2 > 0 \nonumber \\
&Z_{ij} \in \{0,1,2\},\ (i,j) \in \mathcal{P},
\label{eq:zz}
\end{align}
where $\delta(x)$ is the discrete Delta function equal to $1$ if and only if its argument is equal to $0$. The formulation in \eref{eq:zz} allows to derive the complexity of the optimal solution as proportional to $\eta^{\frac{M^2-M-\overline{P}}{2}}$, where $1 < \eta \leq 3$ is the complexity constant that depends on the algorithm used to solving the weighted Steiner tree problem and where the extreme case $\eta=3$ reduces to the exhaustive search.

In order to characterize the complexity of the relaxed optimization problem \eref{Approximate_optimization_problem}, we first identify the number of vertices in the planning graph. As the number of vertices depends on the relative position of the nodes, the data rates, and the reliability functions, this subsection characterizes the worst case complexity. Let $n$ be the maximum number of neighbours of the nodes. From constraint \eref{eq:3}, in each cluster, the number of vertices is bounded by $3^n$. Therefore, the number of vertices of the whole planning graph is bounded by $3^nM$. The formulation of the problem as a graph theory problem allows to derive the complexity of the approximate solution as proportional to $\xi^{3^nM}$, where $1 < \xi \leq 2$ is the complexity constant that depends on the algorithm used to solving the maximum weight clique problem and where the extreme case $\xi=2$ reduces to the exhaustive search.

For the minimal number of neighbours as defined in \eref{eq:nei}, the number of neighbours does not depend on the number of nodes $M$ in the network. Therefore, for a large enough number of base-station, the realistic assumption about the cost of the OF and hybrid RF/FSO links the backhaul network design problem become more mathematically tractable with a reduction in complexity from an order of $2^{M^2}$ to a complexity of order $2^{M}$.

\section{Simulation Results}\label{sec:sim}

This section shows the performance of the proposed solution to the backhaul network planning problem using hybrid RF/FSO technology. The base-stations are randomly placed on a $5$ Km long square. The pre-deployed OF links are randomly placed between the base-stations. The ratio of pre-deployed links by the total number of possible connections is of $1/5$. These simulation assume that the price, the provided data rate, and the reliability are sole function of the distance separating the two end nodes. The cost of a multi-mode OM$3$ $(50/125)$ OF link is, according to various constructors (Asahi Kasei, Chromis, Eska, OFS HCS) between $3$ \$ and $30$ \$ per meter depending on the number of cores. In these simulations, the cost of the optical transceivers, being negligible, is ignored, and a medium price $\pi^{(O)}=13.5$ \$ per meter is adopted. The cost of a hybrid RF/FSO link is taken to be independent of the distance. Given the prices offered by the different constructors (fSONA, LightPointe, and RedLine), two types of costs are considered: $\pi^{(h)}=10$ k\$ and $20$ k\$. The price $\pi^{(h)}=40$ k\$ is proposed as a cut-off price for which hybrid RF/FSO do not represent any advantage.

The data rate of a hybrid RF/FSO links is taken to be $D_t$ over a distance $d_D$ after which it decays exponentially. In other words, $D^{(h)}(x)= D_t$ if $x<d_D$ and $D^{(h)}(x) = D_t\exp^{-(x-d_D)}$ otherwise. The reliability follows a similar model. The maximal distance satisfying the reliability condition is $d_R$. For illustration purposes, the length $d_D$ and $d_R$ are assumed to be $3$ and $2$ Km unless indicated otherwise.

The numbers of base-stations, the price of the hybrid RF/FSO transceivers, and the distances $d_D$ and $d_R$ vary in the simulations so as to study the methods performance for various scenarios. The planning simulated in this section are the optimal planning (solution of \eref{Original_optimization_problem}), the OF only planning (solution of \eref{eq:11}), and our proposed heuristic hybrid RF/FSO-OF planning (solution of \eref{Approximate_optimization_problem}).

\begin{figure}[t]
\centering
\includegraphics[width=0.8\linewidth]{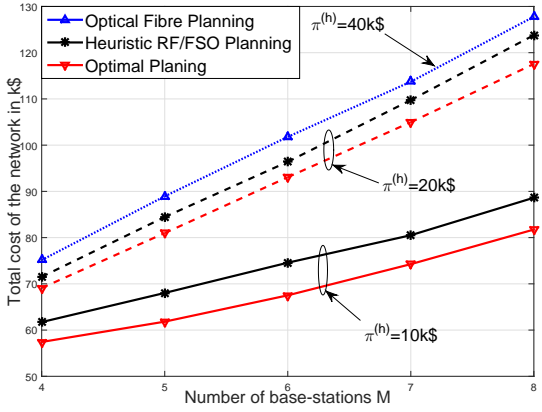}\\
\caption{Mean cost of the network versus the number of base-stations $M$. The solid lines are obtained for a price of a hybrid RF/FSO links of $\pi^{(h)}=10$ k\$, the dashed for a cost $\pi^{(h)}=20$ k\$ and the dotted for the cutoff price $\pi^{(h)}=40$ k\$ at which the different planning coincide.}\label{fig:MC}
\end{figure}

\begin{figure}[t]
\centering
\includegraphics[width=0.8\linewidth]{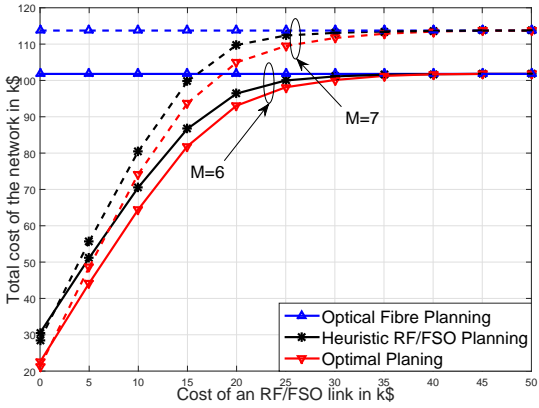}\\
\caption{Mean cost of the network versus the cost of hybrid RF/FSO links $\pi^{(h)}$. The solid lines are obtained for some base-station $M=6$, and the dashed for $M=7$.}\label{fig:CC}
\end{figure}

\begin{figure}[t]
\centering
\includegraphics[width=0.8\linewidth]{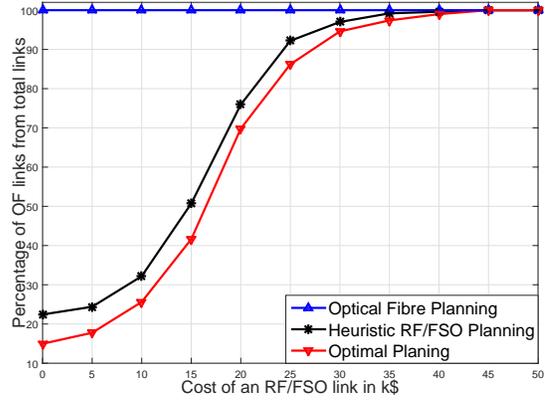}\\
\caption{Average percentage of OF connections versus the cost of hybrid RF/FSO links $\pi^{(h)}$ for a network containing $7$ nodes.}\label{fig:CR}
\end{figure}

\begin{figure}[t]
\centering
\includegraphics[width=0.8\linewidth]{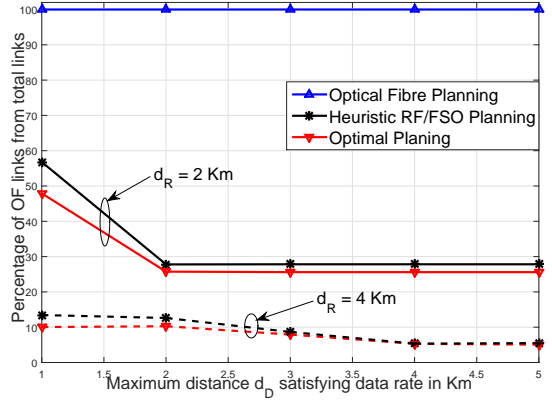}\\
\caption{Average percentage of OF connections versus the distance $d_D$ satisfying the data rate. The solid lines are obtained for a perfect reliability distance $d_R=2Km$ and the dashed one for $d_R=4Km$.}\label{fig:DDR}
\end{figure}

\begin{figure}[t]
\centering
\includegraphics[width=0.8\linewidth]{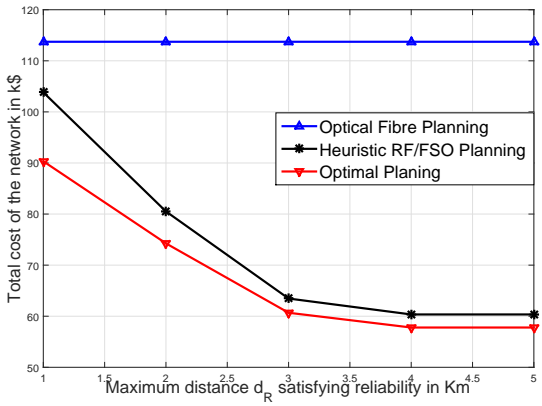}\\
\caption{Mean cost of the network versus the maximum distance $d_R$ satisfying the target reliability $\alpha$ for a system containing $7$ nodes.}\label{fig:DRC}
\end{figure}

\begin{figure}[t]
\centering
\includegraphics[width=0.8\linewidth]{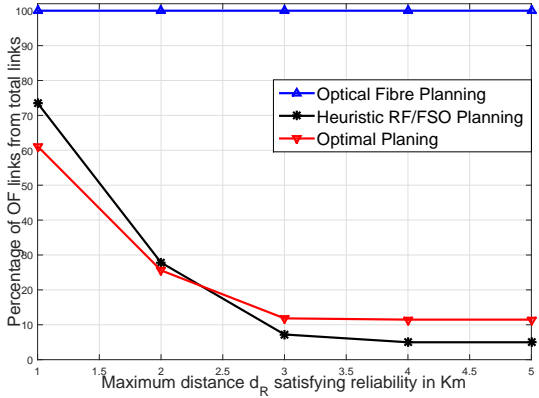}\\
\caption{Average percentage of OF connections versus the maximum distance $d_R$ satisfying the target reliability $\alpha$ for a network containing $7$ nodes.}\label{fig:DRR}
\end{figure}

\begin{figure}[t]
\centering
\includegraphics[width=0.8\linewidth]{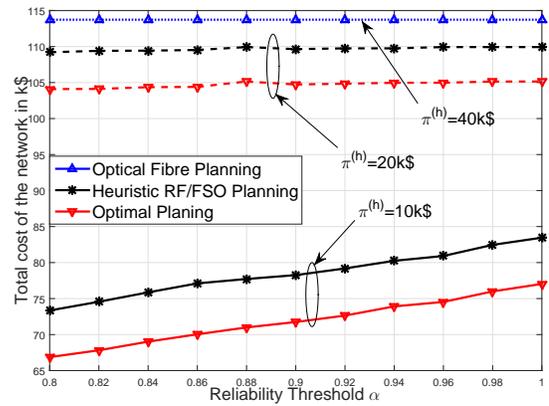}\\
\caption{Mean cost of the network versus the targeted reliability $\alpha$. The solid lines are obtained for a price of a hybrid RF/FSO links of $\pi^{(h)}=10$ k\$, the dashed for a cost $\pi^{(h)}=20$ k\$ and the dotted for the cutoff price $\pi^{(h)}=40$ k\$ at which the different planning coincide.}\label{fig:AC}
\end{figure}

\begin{figure}[t]
\centering
\includegraphics[width=0.8\linewidth]{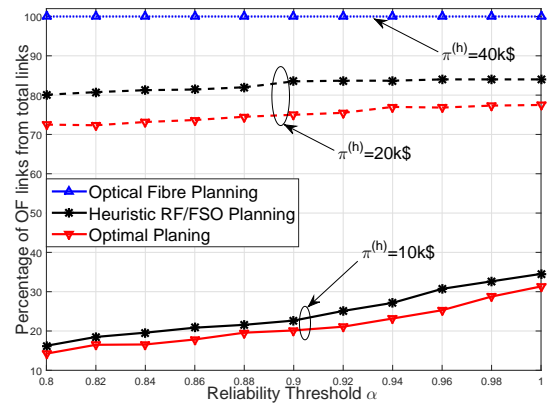}\\
\caption{Average percentage of OF connections versus the targeted reliability $\alpha$. The solid lines are obtained for a price of a hybrid RF/FSO links of $\pi^{(h)}=10$ k\$, the dashed for a cost $\pi^{(h)}=20$ k\$ and the dotted for the cutoff price $\pi^{(h)}=40$ k\$ at which the different planning coincide.}\label{fig:AR}
\end{figure}

\fref{fig:MC} plots the cost of the network versus the number of BSs, for various costs of the hybrid RF/FSO transceivers. We clearly see that the degradation of our proposed solution against the optimal solution becomes less severe when first the number of base-stations increases, and secondly when the hybrid RF/FSO transceivers become more expensive. The increase in performance in the first case can be explained by the fact that the connectivity opportunities of nodes increase as the number of base-stations increases, due to the rise in the neighbours sets $\mathcal{N}_i$. The gain in performance when the price of the hybrid RF/FSO transceivers increases can be explained by the fact that our assumption \eref{eq:12} becomes more valid as the price of the hybrid RF/FSO transceivers increases.

\fref{fig:CC} and \fref{fig:CR} illustrate the cost of the network and the ratio of the OF link, respectively, against the expense of the hybrid RF/FSO transceivers. As shown in \fref{fig:MC}, the performance of our proposed algorithm is more close to the one of the optimal planning as the cost of the hybrid RF/FSO transceivers increases. From \fref{fig:CR}, we clearly see that if the hybrid RF/FSO transceivers are expensive enough, both the optimal and our proposed solution contain only OF links. In fact, for expensive hybrid RF/FSO transceivers, the OF links offers a noticeable rate advantage that explain their use. It is worth mentioning that for a cost $\pi^{(h)} \geq 30$ k\$ in \fref{fig:CR}, even though the link's nature utilized in the optimal solution and our proposed solution are different, the total cost of the network is almost the same (\fref{fig:CC} for $M=7$).

To quantify the performance of the proposed algorithms with respect to the distance $d_D$, \fref{fig:DDR} plots the percentage of the OF link used against the distance $d_D$ for different reliability $d_R$ and a price $\pi_{(h)}=20$ k\$. \fref{fig:DDR} depicts that for a small $d_D$, our proposed solution uses more OF links than the optimal solution. Whereas for a $d_D \geq 2$ the ratio is almost the same. This can be explained by our choice of neighbours $\mathcal{N}_i$. The connectivity opportunities of our proposed solution are less than the one of the optimal solution. Hence, for small $d_D$, to satisfy the rate constraint our proposed solution connects to the neighbours using OF links since the hybrid RF/FSO links do not satisfy the constraint. The optimal solution connects to more nodes (outside the neighbours sets) to meet the rate constraint. We further note that for $d_R=2$ Km, the improvement in the provided data rate of the hybrid RF/FSO link does not decrease the total cost of the network as the used connections are the same. This can be explained by the fact that the solution is limited by the reliability constraint. Therefore, there is no gain in improving the provided data rate of the hybrid RF/FSO link unless the reliability of the connection is improved simultaneously.

\fref{fig:DRC} and \fref{fig:DRR} show the total cost of the network and the percentage of used OF links, respectively, against the maximum distance satisfying the reliability constraint $d_R$ for a system composed of $M=7$ base-stations and for different prices of hybrid RF/FSO of $\pi^{(h)}=20$ k\$. From \fref{fig:DRC}, as the reliability of the hybrid RF/FSO link increases, the proposed solution provides a cost similar to the optimal one. This can be explained by the fact that the reliability condition can be satisfied by connecting to base-stations inside the set of neighbours using exclusively hybrid RF/FSO links. The analysis is more corroborated by \fref{fig:DRR} that shows that the proposed solution uses more and more hybrid RF/FSO links as the reliability of such links increases. Hence, even if the optimal and heuristic hybrid FR/FSO plans are different for high reliability, the total deployment cost of the network is very similar.

Finally, to quantify the performance of the proposed solution against the target reliability,\fref{fig:AC} and \fref{fig:AR} plot the total cost of the network and the percentage of used OF links, respectively, against the reliability threshold $\alpha$ for a system composed of $M=7$ base-stations and for different prices of hybrid RF/FSO links. Again, from \fref{fig:AR} we clearly see that for expensive hybrid RF/FSO transceivers both the optimal solution and our proposed solution use exclusively OF links. \fref{fig:AC} shows that, even for cheap hybrid RF/FSO transceivers, our proposed solution performs as good as the optimal solution for a reliability $\alpha \geq 0.8$ even if the nature of the used links is not the same for that price as displayed in \fref{fig:AR}.

\section{Conclusion}\label{sec:conc}

In this paper, we consider the problem of backhaul network design using the OF and hybrid RF/FSO technologies. We first formulate the planning problem under connectivity, rate constraints, and reliability. We, then, solve the problem optimally when only OF links are allowed. Using the solution of the OF deployment, we formulate an approximation of the general planning problem and show that under a realistic assumption about the relative cost of the OF links and the hybrid RF/FSO transceivers, the solution can be expressed as a maximum weight clique in the planning graph. Simulation results show that our approach shows a close-to-optimal performance, especially for practical prices of the hybrid RF/FSO. As a future research direction, network design can be investigated while taking into account the varying reliability the hybrid RF/FSO links.

\appendices

\numberwithin{equation}{section}

\section{Proof of \lref{l1}}\label{ap1}

To show this lemma, this section express the objective function and the system constraints C1 to C5. The objective function can be written as:
\begin{align}
\cfrac{1}{2} \sum_{i=1}^M \sum_{j=1}^M X_{ij}\pi^{(O)}_{ij} + Y_{ij}\pi^{(h)}_{ij}.
\label{eqobj}
\end{align}

\begin{remark}
Naturally, the objective function should not include the price of the pre-deployed OF links. In other words, it should be written as:
\begin{align}
\cfrac{1}{2} \sum_{i=1}^M \sum_{j=1}^M (X_{ij}-P_{ij})\pi^{(O)}_{ij} + Y_{ij}\pi^{(h)}_{ij}.
\end{align}
However, as the term $- 1/2\sum_{i=1}^M \sum_{j=1}^M P_{ij}\pi^{(O)}_{ij}$ is constant with respect to the optimization variables $X_{ij}$ and $Y_{ij}$, then it is removed throughout this paper (including the simulations).
\end{remark}

The pre-deployed OF connections condition, i.e., constraint C1, states that the planning solution should include the pre-deployed OF links. In other words, if nodes $b_i$ and $b_j$ have a pre-deployed OF link $P_{ij}=1$, then the solution should have the same connection $X_{ij}=1$. However, the non-existence of a pre-deployed OF link does not add extra constraints to the system. Therefore, constraint C1 can be written, for arbitrary nodes $b_i$ and $b_j$, as follows:
\begin{align}
X_{ij}P_{ij}=P_{ij}.
\label{eqc1}
\end{align}
The system constraint C2 implies that, at maximum, only one type of connection may exist between any arbitrary nodes $b_i$ and $b_j$. Hence, constraint C2 can be mathematically written as follows:
\begin{align}
X_{ij}Y_{ij} = 0.
\label{eqc2}
\end{align}
For an arbitrary BS $b_i$, the data rate constraint C3 is satisfied if the sum of the data rate provided by all adjacent nodes exceeds the targeted data rate. Therefore, the constraint can be formulated for all BS $b_i$ as follows:
\begin{align}
&\sum_{j=1}^M X_{ij}D^{(O)}_{ij}+ Y_{ij}D^{(h)}_{ij} \geq D_t.
\end{align}
As OF links always satisfy the targeted data rate, i.e., $D^{(O)}_{ij} \geq D_t, \ \forall \ i \neq j$, then the data rate constraint C3 for node $b_i$ can be reformulated as follows:
\begin{align}
&\sum_{j=1}^M X_{ij}D_t + Y_{ij}D^{(h)}_{ij} \geq D_t.
\label{eqc3}
\end{align}

The reliability constraint C4 implies that each node should be connected to the network, at all time, with probability $\alpha$. As the reliability of each link is independent of the other links, such constraint can be formulated for node $b_i$ using the complementary event as follows:
\begin{align}
1 - \prod_{j=1}^M (1 - X_{ij} R^{(O)}_{ij})(1 - Y_{ij} R^{(h)}_{ij}) \geq \alpha.
\end{align}
As OF links are always reliable, i.e., $R^{(O)}(d_{ij}) \geq \alpha \ \forall i \neq j$, the reliability constraint can be simplified as follows:
\begin{align}
1 - \prod_{j=1}^M (1 - X_{ij}\alpha)(1 - Y_{ij} R^{(h)}_{ij}) \geq \alpha.
\label{eqc4}
\end{align}
Define $\mathbf{C} = [c_{ij}]$ as the adjacency matrix by $c_{ij} = X_{ij}+Y_{ij}$. Since only one type of connections exists between the same BSs then $c_{ij}$ is a binary variable (i.e., $c_{ij} \in \{0,1\}$). The connectivity constraint C5 implies that the graph representing the base-stations is connected. From a graph theory perspective \cite{25181258}, such graph connectivity constraint can be expressed as a function of the Laplacian matrix $L$ defined as $\mathbf{L} = \mathbf{D}-\mathbf{C}$, where $\mathbf{D}=\text{diag}(d_1,\ \cdots,\ d_M)$ is a diagonal matrix with $d_i = \sum_{j=1}^M c_{ij}$. The diagonalization of the Laplacian matrix is given by $\mathbf{L}= \mathbf{Q}\mathbf{\Lambda} \mathbf{Q}^{-1}$, where $\mathbf{\Lambda}= \text{diag}(\lambda_1,\ \cdots,\ \lambda_M)$ with $\lambda_1 \leq \lambda_2 \leq \cdots \leq \lambda_M$. The connectivity condition C5 of the matrix can be written using the algebraic formulation proposed in \cite{25181258} as:
\begin{align}
\lambda_2 > 0.
\label{eqc5}
\end{align}

Combining the objective function \eref{eqobj} and the constraints \eref{eqc1}, \eref{eqc2}, \eref{eqc3}, \eref{eqc4}, and \eref{eqc5}, gives the optimization problem proposed in \eref{Original_optimization_problem}.

\section{Proof of \thref{th1}}\label{ap3}

To show the theorem, this section proposes first to demonstrate that \algref{alg1} outputs the optimal solution to the optimization problem stated in \lref{l2} for a network without pre-deployed links. The second part of the section extends the result to network with pre-deployed OF connections.

\subsection{Network Without Pre-deployed OF Links}

To proof this theorem for a network without pre-deployed links, we first prove that \algref{alg1} produces a feasible solution to the problem. Afterward, we show that any graph that can be reduced, using \algref{alg2}, to a single cluster includes the graph outputted by \algref{alg1}. Finally, we show that any solution that cannot be reduced to a single cluster is not optimal.

\algref{alg2} can be seen as a complement of \algref{alg1}. As for \algref{alg1}, in \algref{alg2} begins by generating a cluster of each BS in the system. Afterward, two groups at the minimum price of each other and whose BSs at the minimum price are connected are merged into a single cluster. The process is repeated until no further connection can be found. The steps of the algorithm are summarized in \algref{alg2}.

Let $\overline{\mathcal{Z}}= \{\overline{Z}_1,\ \cdots,\ \overline{Z}_{|\overline{\mathcal{Z}}|}\}$ be the clustering at any step of \algref{alg1}. First note that $\overline{\mathcal{Z}}$ is a partition of $\mathcal{B}$. We proof by induction that nodes inside any cluster $\overline{Z}_i,\ 1 \leq i \leq |\overline{\mathcal{Z}}|$ are connected. Clearly, for a cluster $\overline{Z}_i$ with $|\overline{Z}_i|=1$ (the cluster contains a single node), all nodes inside the cluster are connected. Assume that all clusters $\overline{Z}_i$ of size $|\overline{Z}_i| \leq n$ are connected. From the last step of \algref{alg1} clusters of size $n+1$ can be generated only by merging two clusters $\overline{Z}_j$ and $\overline{Z}_k$ with $|\overline{Z}_j|,|\overline{Z}_k| < n$ and $|\overline{Z}_j|+|\overline{Z}_k|=n+1$. Since by construction such clusters are connected ($\overline{X}_{jk}=\overline{X}_{kj}=1$ with $b_j \in \overline{Z}_j$ and $b_k \in \overline{Z}_k$), then the resulting cluster $\overline{Z}_i$ from merging $\overline{Z}_j$ and $\overline{Z}_k$ is also connected. Therefore, all nodes within any arbitrary cluster $\overline{Z}_i,\ 1 \leq i \leq |\overline{\mathcal{Z}}|$ are connected. Finally, since $\overline{\mathcal{Z}}$ contains a single cluster at the end of \algref{alg1} (i.e., $|\overline{\mathcal{Z}}|=1$) and it is a partition of $\mathcal{B}$ (contains all the nodes in the network), then all the nodes are connected. Hence the outputted solution satisfy constraint \eref{eq:8}. By construction, we can easily see that the connections are binary and symmetric. In other words, the outputted solution satisfy constraints \eref{eq:9} and \eref{eq:10} which conclude that it is a feasible solution.

\begin{algorithm}[t]
\begin{algorithmic}
\REQUIRE $\mathcal{B}, \pi^{(O)},$ and $X_{ij},\ 1 \leq i,j\leq M$.
\STATE Initialize $\mathcal{Z} = \varnothing$.
\FORALL {$b \in \mathcal{B}$}
\STATE $\mathcal{Z} = \{\mathcal{Z},\{b\}\}$.
\ENDFOR
\STATE Initialize $t=$TRUE.
\WHILE {$t=$TRUE}
\STATE $t=$FALSE.
\FORALL {$Z \neq Z^{\prime} \in \mathcal{Z}$}
\IF {$\sum\limits_{\substack{b_i \in Z \\ b_j \in Z^{\prime}}}X_{ij} = 1$}
\STATE $Z^* = \arg \min\limits_{\substack{X \in \mathcal{Z} \\ X \neq Z^{\prime}}} \left[ \min\limits_{\substack{b \in X \\ b^{\prime} \in Z^{\prime}}} \pi^{(O)}(b , b^{\prime})\right]$
\STATE $Z^{\prime *} = \arg \min\limits_{\substack{X \in \mathcal{Z} \\ X \neq Z}} \left[ \min\limits_{\substack{b \in X \\ b^{\prime} \in Z}} \pi^{(O)}(b^{\prime} ,b)\right]$.
\IF {$Z = Z^*$ and $Z^{\prime}=Z^{\prime*}$}
\STATE $(b_i,b_j) = \arg \min\limits_{\substack{b \in Z \\ b^{\prime} \in Z^{\prime}}} \pi^{(O)}(b , b^{\prime})$.
\IF {$X_{ij}=1$}
\STATE $\mathcal{Z} = \mathcal{Z} \setminus \{Z\}$
\STATE $\mathcal{Z} = \mathcal{Z} \setminus \{Z^{\prime}\}$
\STATE $\mathcal{Z} = \{\mathcal{Z},\{Z_i,Z_j\}\}$.
\STATE $t=$TRUE
\ENDIF
\ENDIF
\ENDIF
\ENDFOR
\ENDWHILE
\end{algorithmic}
\caption{Clustering Algorithm}
\label{alg2}
\end{algorithm}

Let $X_{ij}, 1 \leq i,j \leq M$ be a feasible solution and let $\mathcal{Z}$ be the outputted clustering by \algref{alg2} when inputted $X_{ij}$. We can clearly see that $\mathcal{Z}$ is a partition of $\mathcal{B}$. Therefore, if $|\mathcal{Z}|=1$, then $\mathcal{Z}=\overline{\mathcal{Z}}$ since there exist only a unique partition containing a single element (the set $\mathcal{B}$ itself). This concludes that any graph that can be reduced, using \algref{alg2}, to a single cluster includes all the connections that are created in the graph outputted by \algref{alg1}. Since $\pi^{(O)}$ is a strictly positive function and that the graph outputted by \algref{alg1} have the minimum number of connections among all the graphs that can be reduced to a single cluster using \algref{alg2}, then the solution of \algref{alg1} is the best solution among the solutions that can be reduced to a single cluster using \algref{alg2}.

Now assume that $|\mathcal{Z}| \neq 1$. We can clearly see that $|\mathcal{Z}| \geq 3$. Otherwise, if there exist only two clusters (i.e., $|\mathcal{Z}| =2$) and since the solution is feasible, then they are connected. Due to the fact they are only two, then two cases can be distinguished:
\begin{itemize}
\item The BSs at the minimum price of each others are connected and hence they can be reduced to a single cluster. Therefore, clustering $|\mathcal{Z}| =2$ cannot be outputted by \algref{alg2}.
\item The BSs at the minimum price of each others are not connected. Then the solution having the same connections except for the link between the two cluster being replaced with the connection of BSs at the minimum price of each others produces a feasible solution at a lower cost. Therefore, the initial solution is not optimal.
\end{itemize}

Let $\mathcal{Z}= \{Z_1,\ \cdots,\ Z_{|\mathcal{Z}|}\}, |\mathcal{Z}| \geq 3$ be the outputted clustering by \algref{alg2}. Define the reduced clustering as $\tilde{\mathcal{Z}}= \{Z_1,\ Z_2,\ \{Z_3 \cup \cdots \cup \ Z_{|\mathcal{Z}|}\}\ \} = \{\tilde{Z_1},\ \tilde{Z_2},\ \tilde{Z_3}\}$. We can clearly see that $|\tilde{\mathcal{Z}}| = 3$ with none of the clusters connected and at minimum price of each other.

\begin{lemma}
For any three points in the plane, there must exist two points at the minimum distance from each others.
\label{l3}
\end{lemma}

\begin{proof}
Let $a,b,c$ be the three points in the plane and assume that there do not exist two points at minimum distance of each other. The only possible configuration (up to a permutation of the points) is that $a$ at minimum distance from $b$ and $b$ is not. Hence $b$ is at minimum distance from $c$ which is at its turn at minimum distance from $a$. These conditions yield $d(a,b)<d(a,c)$, $d(b,c)<d(b,a)$ and $d(c,a)<d(c,b)$. Since the distance operator is symmetric then, $d(c,a)<d(c,b)=d(b,c)<d(b,a)=d(a,b)<d(a,c)$. In other words, $d(a,c)<d(a,c)$, which is impossible. Therefore for any three points in the plane, there must exist two points at minimum distance of each others.
\end{proof}

From \lref{l3}, there must exist two clusters at the minimum price of each other. Since they have not been reduced to a single cluster, then they are not connected. For simplicity, assume $\tilde{Z_1}$ and $\tilde{Z_2}$ are at minimum price of each other and since the graph is connected then $\tilde{Z_1}$ and $\tilde{Z_3}$ are connected and similarly for $\tilde{Z_2}$ and $\tilde{Z_3}$. Moreover, it can be easily concluded that all nodes inside the clusters are connected. Otherwise, assume $\tilde{Z_3}$ can be split into two non-connected clusters $\tilde{Z}$ and $\tilde{Z}^{\prime}$ with $\tilde{Z_1}$ connected only to $\tilde{Z}$ and $\tilde{Z_2}$ connected only to $\tilde{Z}^{\prime}$. Then, since $\tilde{Z_1}$ and $\tilde{Z_2}$ are not connected, the whole graph is not connected and the solution is not feasible.

The clustering connecting $\tilde{Z_1}$ with $\tilde{Z_2}$ and $\tilde{Z_1}$ with $\tilde{Z_3}$ (or $\tilde{Z_2}$ with $\tilde{Z_3}$) produces also a feasible solution at a lower cost since the sum of the prices is minimized ($\pi^{(O)}(\tilde{Z_1},\tilde{Z_2}) < \min(\pi^{(O)}(\tilde{Z_1},\tilde{Z_3}),\pi^{(O)}(\tilde{Z_2},\tilde{Z_3}))$). Therefore, $\tilde{\mathcal{Z}}$ is not optimal and by extension $\mathcal{Z}$ is also not optimal. Finally, we can conclude that the optimal solution is the solution containing a single cluster. Therefore, the solution outputted by \algref{alg1} is the optimal solution to the problem proposed in \lref{l2}.

\subsection{Network With Pre-deployed OF Links}

It can explicitly be noted that for a network with pre-deployed OF links, \algref{alg1} produces a feasible solution. In fact, as for the previous subsection, the solution outputted by the algorithm satisfy the constraints \eref{eq:8}, \eref{eq:9}, and \eref{eq:10}. Furthermore, from the initialization of the variables $X_{ij}, 1 \leq i,j \leq M$, the solution satisfies constraint \eref{eq:np2}.

From the initialization of the clusters, all the base-stations that are connected with pre-deployed OF links are merged in the same cluster. Such initialization implies the following two properties:
\begin{enumerate}
\item Only the pre-deployed OF links connect such BSs.
\item Any node connected to the cluster is connected to the BS with minimal cost inside that cluster.
\end{enumerate}
Therefore, to show that the solution of \algref{alg1} is the optimal planning for the problem proposed in \lref{l2}, it is sufficient to show that any solution that violates the properties above is not optimal. Finally, as clusters can be seen as new nodes in a network without pre-deployed OF links, the result of the first part of the section guarantee the optimality of the solution.

Assume that in the optimal solution, a connection exists between two nodes in the same cluster. As these nodes are connected via pre-deployed OF link through single or multi-hop connection, then removing the extra connection produces a feasible solution at a lower cost. Hence, the optimal solution contains only the pre-deployed OF links connecting BSs in each cluster.

Similarly, assume there exists a node connected to a cluster in the optimal solution that is not linked to the BS with minimal cost in that cluster. It can be readily seen that as the nodes in the cluster are all connected, then removing the connection and replacing it with the minimal cost one produces a feasible solution at a lower cost. Therefore, such connection does not exist in the optimal solution. Finally, the solution generated by \algref{alg1} is the optimal planning for the problem proposed in \lref{l2}.

\section{Proof of \lref{l4}}\label{ap4}

To prove this lemma, we prove that any solution to \eref{Approximate_optimization_problem} is a feasible solution to \eref{Original_optimization_problem}. We first show that constraint \eref{eq:nr2} is equivalent to constraint \eref{eq:nr}. In the original problem formulation \eref{Original_optimization_problem}, the reliability constraint is:
\begin{align}
1 - \prod_{j=1}^M (1 - X_{ij}\alpha)(1 - Y_{ij} R^{(h)}_{ij}) \geq \alpha.
\label{eq:ap41}
\end{align}

It can easily be seen that \eref{eq:ap41} is satisfied, for node $b_i$ if and only if at least one of the following options is correct:
\begin{enumerate}
\item $\sum_{j=1}^M X_{ij} > 0 \Rightarrow$ Node $b_i$ is connected to another node with an OF link which provides full reliability.
\item $\sum_{j \in \mathcal{R}_i}^M Y_{ij} > 0 \Rightarrow$ Node $b_i$ is connected to another node with a hybrid RF/FSO link to a nearby node that provides full reliability.
\item $1 - \prod_{j \in \overline{\mathcal{R}}_i}(1 - Y_{ij} R^{(h)}_{ij}) \geq \alpha \Rightarrow$ Node $b_i$ is connected to a sufficiently large number of nodes to satisfy the reliability constraint.
\end{enumerate}

In the rest part of the proof, we show that constraint \eref{eq:nr2} regroups the three scenarios mentioned above. Applying a logarithmic transformation to the third option and rearranging the terms yields the following expression:
\begin{align}
\sum_{j \in \overline{\mathcal{R}}_i}\log (1-Y_{ij} R^{(h)}_{ij}) \leq \log(1-\alpha).
\end{align}
Given that the reliability of nodes $j \notin \mathcal{R}_i$ is small, i.e., $R^{(h)}_{ij} \lll 1$, then $Y_{ij} R^{(h)}_{ij} \lll 1$, $\forall \ j \notin \mathcal{R}_i$. Applying a first order Taylor expansion of the logarithm yield the condition:
\begin{align}
\sum_{j \in \overline{\mathcal{R}}_i} -Y_{ij} R^{(h)}_{ij} \leq \log(1-\alpha).
\label{eq:ap42}
\end{align}

Let $\tilde{\alpha} = \cfrac{1}{\log(1-\alpha)}$. Rearranging the terms of \eref{eq:ap42} and including the first and second option discussed in the previous paragraph gives the following constraint:
\begin{align}
\sum_{j=1}^M X_{ij}\tilde{\alpha} + \sum_{j \in \mathcal{R}_i}Y_{ij}\tilde{\alpha} + \sum_{j \in \overline{\mathcal{R}}_i}Y_{ij} R^{(h)}_{ij} \geq \tilde{\alpha}.
\label{eq:ap43}
\end{align}

It can clearly be seen that setting $\sum_{j=1}^M X_{ij} > 0$, and/or $\sum_{j \in \mathcal{R}_i}^M Y_{ij} > 0$, and/or $1 - \prod_{j \in \overline{\mathcal{R}}_i}(1 - Y_{ij} R^{(h)}_{ij}) \geq \alpha$ automatically satisfy \eref{eq:ap43} which proves that constraint \eref{eq:nr2} is equivalent to constraint \eref{eq:nr}.

We, now, show that constraint \eref{eq:14} is included in constraint \eref{eq:6} as it is the second and last constraint changing from one formulation to the other. Constraint \eref{eq:14} ensures that, for all connections $\overline{X}_{ij}=1$ that are generated by \algref{alg1}, a similar connections (OF or hybrid RF/FSO link) between nodes $b_i$ and $b_j$ must exist. For connections $\overline{X}_{ij}=0$, the constraint is always satisfied and connection may or may not exist. From \thref{th1}, \algref{alg1} produces a connected graph. In other words, $\lambda_2 > 0$. Therefore, constraint \eref{eq:14} is included in constraint \eref{eq:6}. A feasible solution to \eref{Approximate_optimization_problem} is, therefore, a feasible solution to \eref{Original_optimization_problem}. In \thref{th2}, we show that the optimal solution to \eref{Approximate_optimization_problem} is the optimal solution to \eref{Original_optimization_problem} in many scenarios (but not all). Therefore, the approximation of problem \eref{Original_optimization_problem} by the problem \eref{Approximate_optimization_problem} is tight.

\section{Proof of \thref{th2}}\label{ap5}

In this theorem, we show that the optimal solution $X^*_{ij},Y^*_{ij},\ 1 \leq i,j \leq M$ to \eref{Approximate_optimization_problem} should satisfy $X^*_{ij}+Y^*_{ij}=1$ only if $(i,j) \in \mathcal{N}_j \times \mathcal{N}_i$.

\begin{remark}
Note that if any feasible solution $X_{ij},Y_{ij}$ to the general problem \eref{Original_optimization_problem} that verify \thref{th2}, then the solution is feasible to \eref{Approximate_optimization_problem}. In other words, if $X_{ij},Y_{ij}$ feasible to \eref{Original_optimization_problem} and $X_{ij}+Y_{ij}=1$ only if $(i,j) \in \mathcal{N}_j \times \mathcal{N}_i$, then $(X_{ij}+Y_{ij})\overline{X}_{ij}=\overline{X}_{ij}$. This can be easily concluded given the construction of $\mathcal{N}_i, \leq i \leq M$ as the minimum set of nodes that can generate a connected graph. Since all $(x,y)$ such that $\overline{X}_{xy}=1$ are at the edge of at least one of the $\mathcal{N}_i$, then the only connected solution that satisfies \thref{th2} is feasible to \eref{Approximate_optimization_problem}. In that scenario, the optimal solution of \eref{Original_optimization_problem} and \eref{Approximate_optimization_problem} are the same.
\end{remark}

Assume $\exists (x,y)$ such that $X^*_{xy}+Y^*_{xy}=1$ and $(x,y) \notin \mathcal{N}_y \times \mathcal{N}_x$. Two scenarios can be distinguished:
\begin{itemize}
\item $\mathcal{B} \setminus \{b_x,b_y\}$ represents a connected subgraph.
\item $\mathcal{B} \setminus \{b_x,b_y\}$ is not a connected subgraph.
\end{itemize}

For the first scenario, consider the reduced network $\{b_x,b_y,\tilde{b}_k\}$. Clearly, we have $b_{x^*},b_{y^*} \in \tilde{b}_k$. Define the following planning:
\begin{align}
\tilde{X}_{ij} &=
\begin{cases}
1 \hspace{1cm} &\text{if } i=x \text{ and } j=x^* \\
1 \hspace{1cm} &\text{if } j=y \text{ and } j=y^* \\
0 \hspace{1cm} &\text{if } i=x \text{ and } j=y \\
X^*_{ij} \hspace{1cm} &\text{otherwise}
\end{cases} \nonumber \\
\tilde{Y}_{ij} &=
\begin{cases}
0 \hspace{1cm} &\text{if } i=x \text{ and } j=x^* \\
0 \hspace{1cm} &\text{if } j=y \text{ and } j=y^* \\
0 \hspace{1cm} &\text{if } i=x \text{ and } j=y \\
Y^*_{ij} \hspace{1cm} &\text{otherwise},
\end{cases}
\end{align}
in which the connection between $b_x$ and $b_y$ is replaced by two connections between $b_x$ and $b_{x^*}$ and between $b_y$ and $b_{y^*}$. We can clearly see that the network is connected. Moreover, since $b_x$ and $b_y$ are connected with an OF link, then the data rate constraint is satisfied for both nodes. Therefore, $\tilde{X}_{ij},\tilde{Y}_{ij},\ 1 \leq i,j \leq M$ represents a feasible solution. Moreover, The difference in cost between the optimal planning $X^*_{ij},Y^*_{ij}$ and the planning $\tilde{X}_{ij},\tilde{Y}_{ij}$ is lower bounded by:
\begin{align}
&\pi(b_x,b_y) - (\pi^{(O)}(b_x,b_{x^*}) + \pi^{(O)}(b_y,b_{y^*}))   \\
&\geq \pi^{(h)}(b_x,b_y) - (\pi^{(O)}(b_x,b_{x^*}) + \pi^{(O)}(b_y,b_{y^*})) \geq 0.\nonumber
\end{align}
From assumption \eref{eq:12}, the difference is positive. This concludes that $X^*_{i,j},Y^*_{i,j},\ 1 \leq i,j \leq M$ is not the optimal solution.
\begin{remark}
For scenario 1, $X^*_{i,j},Y^*_{i,j},\ 1 \leq i,j \leq M$ can be the optimal solution to the original problem \eref{Original_optimization_problem}. Hence, for this configuration, the optimal solution of \eref{Original_optimization_problem} and \eref{Approximate_optimization_problem} are the same.
\end{remark}

For scenario 2, let the network be reduced to $\{b_x,b_y,\tilde{b}_k,\tilde{b}_l\}$ with $b_x$ connected to $\tilde{b}_k$, which is a connected subgraph, $b_y$ connected $\tilde{b}_l$, which is a connected subgraph, and $\tilde{b}_k$ and $\tilde{b}_l$ are not connected. Since $X^*_{i,j},Y^*_{i,j},\ 1 \leq i,j \leq M$ is a feasible solution to \eref{Approximate_optimization_problem}, then it satisfies constraint \eref{eq:14}. In other words, $X^*_{k,l} + Y^*_{k,l}=1, \ \forall \ k,l$ such that $\overline{X}_{kl} =1$. Note that $\overline{X}_{xy} = 0$. Otherwise, by construction of the neighbours sets, we have $(x,y) \in \mathcal{N}_y \times \mathcal{N}_x$. Define the planning $\tilde{X}_{ij},\tilde{Y}_{ij}$ such that
\begin{align}
\tilde{X}_{ij} &=
\begin{cases}
0 \hspace{1cm} &\text{if } i=x \text{ and } j=y \\
X^*_{ij} \hspace{1cm} &\text{otherwise}
\end{cases} \nonumber \\
\tilde{Y}_{ij} &=
\begin{cases}
0 \hspace{1cm} &\text{if } i=x \text{ and } j=y \\
Y^*_{ij} \hspace{1cm} &\text{otherwise}.
\end{cases}
\end{align}

The planning $\tilde{X}_{ij},\tilde{Y}_{ij}$ satisfy \eref{eq:14}. However, we can clearly see that the graph is not connected. Therefore, $X^*_{i,j},Y^*_{i,j},\ 1 \leq i,j \leq M$ is not a feasible solution. This concludes that scenario 2 is not feasible. Finally, we conclude that the optimal solution $X^*_{i,j},Y^*_{i,j},\ 1 \leq i,j \leq M$ to \eref{Approximate_optimization_problem} should satisfy $X^*_{xy}+Y^*_{xy}=1$ only if $(x,y) \in \mathcal{N}_y \times \mathcal{N}_x$.
\begin{remark}
Scenario 2 can be a feasible scenario if $X^*_{i,j},Y^*_{i,j},\ 1 \leq i,j \leq M$ is the optimal solution to the original problem \eref{Original_optimization_problem}. In that case, two scenarios can be distinguished:
\begin{itemize}
\item $b_{x^*} \in \mathcal{N}_y$ and $b_{y^*} \in \mathcal{N}_x$. In that case $\tilde{X}_{ij},\tilde{Y}_{ij}$ presented for scenario 1 produces a feasible solution with lower cost. Therefore, the optimal solution of \eref{Original_optimization_problem} and \eref{Approximate_optimization_problem} are the same.
\item $b_{x^*} \notin \mathcal{N}_y$ or $b_{y^*} \notin \mathcal{N}_x$. In this configuration, no conclusion can be reached about the optimal solution of \eref{Original_optimization_problem} and \eref{Approximate_optimization_problem} is an upper bound of the minimum of \eref{Original_optimization_problem}.
\end{itemize}
Therefore, the approximation of the problem \eref{Original_optimization_problem} by the problem \eref{Approximate_optimization_problem} is a tight approximation.
\end{remark}

\section{Proof of \thref{th3}}\label{ap6}

To proof this theorem, we first prove that there is a one to one mapping between the set of feasible solution of a modified version of problem \eref{Approximate_optimization_problem} and the set of cliques of degree $M$ in the \emph{planning} graph $\mathcal{G}(\mathcal{V},\mathcal{E})$. To conclude the proof, we show that the weight of the clique is equivalent to the merit function of the optimization problem \eref{Approximate_optimization_problem}.

From \thref{th2}, we have $X_{ij}=0$ and $Y_{ij}=0,\ \forall \ (b_i,b_j) \notin \mathcal{N}_j \times \mathcal{N}_i $. Hence the objective function \eref{eq:25} and constraint \eref{eq:23} can be replaced by:
\begin{align}
&\max -\cfrac{1}{2} \sum_{i=1}^M \sum_{j \in \mathcal{N}_i} X_{ij}\pi^{(O)}_{ij} + Y_{ij}\pi^{h}_{ij} .
\end{align}

Similarly, since $X_{ij}=0$ and $Y_{ij}=0,\ \forall \ (b_i,b_j) \notin \mathcal{N}_j \times \mathcal{N}_i $, then $(X_{ij},Y_{ij}) = (X_{ji},Y_{ji})$ is always verified for $(b_i,b_j) \notin \mathcal{N}_j \times \mathcal{N}_i $. Hence the constraints \eref{eq:20}, \eref{eq:21}, and \eref{eq:22} can be replaced by:
\begin{align}
&(X_{ij},Y_{ij}) = (X_{ji},Y_{ji}) ,\ \forall \ (b_i,b_j) \in \mathcal{N}_j \times \mathcal{N}_i \nonumber \\
&X_{ij}Y_{ij} = 0,\ \forall \ (b_i,b_j) \in \mathcal{N}_j \times \mathcal{N}_i \nonumber \\
&\sum_{j \in \mathcal{N}_i} X_{ij}D_t + Y_{ij}D^{(h)}_{ij} \geq 1,\ D_t \leq i \leq M.
\end{align}

By definition of the set $\mathcal{N}_i$, we have $\overline{X}_{ij}=0, \forall \ j \notin \mathcal{N}_i$. Therefore, constraint \eref{eq:14} may be written as:
\begin{align}
&(X_{ij} + Y_{ij})\overline{X}_{ij} = \overline{X}_{ij}, \forall \ j \in \mathcal{N}_i, 1 \leq i \leq M.
\end{align}

The problem \eref{Approximate_optimization_problem} can be reformulated as:
\begin{align}
\max &-\cfrac{1}{2} \sum_{i=1}^M \sum_{j \in \mathcal{N}_i} X_{ij}\pi^{(O)}_{ij} + Y_{ij}\pi^{h}_{ij} \nonumber \\
\text{subject to } &(X_{ij},Y_{ij}) = (X_{ji},Y_{ji}) ,\ \forall \ (b_i,b_j) \in \mathcal{N}_j \times \mathcal{N}_i \nonumber \\
&X_{ij}P_{ij} = P_{ij} ,\ \forall \ (b_i,b_j) \in \mathcal{N}_j \times \mathcal{N}_i \nonumber \\
&X_{ij}Y_{ij} = 0,\ \forall \ (b_i,b_j) \in \mathcal{N}_j \times \mathcal{N}_i \nonumber \\
&\sum_{j \in \mathcal{N}_i} X_{ij}D_t + Y_{ij}D^{(h)}_{ij} \geq D_t,\ 1 \leq i \leq M \nonumber \\
&\sum_{j \in \mathcal{N}_i \cap \mathcal{R}_i} (X_{ij}+Y_{ij})\tilde{\alpha} + \sum_{j \in \in \mathcal{N}_i \cap\overline{\mathcal{R}}_i}Y_{ij} R^{(h)}_{ij} \geq \tilde{\alpha} \nonumber \\
&(X_{ij} + Y_{ij})\overline{X}_{ij} = \overline{X}_{ij}, \forall \ j \in \mathcal{N}_i, 1 \leq i \leq M \nonumber \\
&X_{ij},Y_{ij} \in \{0,1\},\ 1 \leq i,j\leq M.
\end{align}

Let $\gamma_i = \{(X_{ij_1},Y_{ij_1}),\ \cdots,\ (X_{ij_{|\mathcal{N}_i|}},Y_{ij_{|\mathcal{N}_i|}})\}$ be the new variable. Using the variables $\gamma_i$ and the definition of the sets $\mathcal{C}_i$, the problem can be written as:
\begin{subequations}
\begin{align}
\max &\sum_{i=1}^M w(\gamma_i) \nonumber \\
\text{subject to } &\gamma_i \in \mathcal{C}_i, 1 \leq i \leq M \label{eq:31} \\
&(X_{ij},Y_{ij}) = (X_{ji},Y_{ji}) ,\ \forall\ 1 \leq i \neq j \leq M. \label{eq:30}
\end{align}
\label{eq:32}
\end{subequations}

Let $\mathbb{C}$ be the set of cliques of degree $M$ in the \emph{planning} graph and let $\mathbb{F}$ be the set of feasible solutions to the optimization problem \eref{eq:32}. We first proof that any clique $C=\{v_1,\ \cdots,\ v_M\} \in \mathbb{C}$ satisfy constraints \eref{eq:31}, and \eref{eq:30}. Then, we prove the converse. In other words, for element in $\mathbb{F}$, there exists a clique in $\mathbb{C}$.

Let $C=\{v_1,\ \cdots,\ v_M\} \in \mathbb{C}$. Assume $\exists \ k,i,j$ such that $v_i,v_j \in \mathcal{C}_k$. Since all the vertices in a clique are connected, then from the connectivity condition C1, vertices $v_i$ and $v_j$ are not connected. Hence $\nexists \ k,i,j$ such that $v_i,v_j \in \mathcal{C}_k$. Given that the clique contain $M$ elements, then constraint \eref{eq:31} is satisfied. The connectivity condition C2 ensures that $(X_{ij},Y_{ij}) = (X_{ji},Y_{ji})$ for all vertices. Therefore, $C$ is a feasible solution to \eref{eq:32}. Similarly, let $\{c_1,\ \cdots,\ c_M\}$ be a feasible solution to \eref{eq:32}, then clearly the vertices corresponding to each cluster are connected. Finally, there a one to one mapping between $\mathbb{C}$ and $\mathbb{F}$. Moreover, the weight of the clique is $w(C) = \sum_{i=1}^M w(\gamma_i) $ which conclude that the solution of \eref{eq:32} is the maximum weight clique, among the clique of size $M$, in the \emph{planning} graph.

\bibliographystyle{IEEEtran}
\bibliography{citations}

\begin{thebibliography}{10}
\providecommand{\url}[1]{#1}
\csname url@samestyle\endcsname
\providecommand{\newblock}{\relax}
\providecommand{\bibinfo}[2]{#2}
\providecommand{\BIBentrySTDinterwordspacing}{\spaceskip=0pt\relax}
\providecommand{\BIBentryALTinterwordstretchfactor}{4}
\providecommand{\BIBentryALTinterwordspacing}{\spaceskip=\fontdimen2\font plus
\BIBentryALTinterwordstretchfactor\fontdimen3\font minus
  \fontdimen4\font\relax}
\providecommand{\BIBforeignlanguage}[2]{{%
\expandafter\ifx\csname l@#1\endcsname\relax
\typeout{** WARNING: IEEEtran.bst: No hyphenation pattern has been}%
\typeout{** loaded for the language `#1'. Using the pattern for}%
\typeout{** the default language instead.}%
\else
\language=\csname l@#1\endcsname
\fi
#2}}
\providecommand{\BIBdecl}{\relax}
\BIBdecl

\bibitem{546843}
J.~G. Andrews, S.~Buzzi, W.~Choi, S.~Hanly, A.~Lozano, A.~C.~.~K. Soong, and
  J.~C. Zhang, ``What will {5G} be?'' \emph{IEEE Journal on Selected Areas in
  Communications}, vol.~32, no.~6, pp. 1065--1082, Jun. 2014.

\bibitem{6736746}
F.~Boccardi, R.~Heath, A.~Lozano, T.~Marzetta, and P.~Popovski, ``Five
  disruptive technology directions for 5{G},'' \emph{IEEE Communications
  Magazine}, vol.~52, no.~2, pp. 74--80, February 2014.

\bibitem{13050958}
R.~Ford, C.~Kim, and S.~Rangan, ``Opportunistic third-party backhaul for
  cellular wireless networks,'' in \emph{Proc. of Asilomar Conference on
  Signals, Systems and Computers (ACSSC' 2013), Pacific Grove, CA, USA},
  November 2013, pp. 1594--1600.

\bibitem{591842}
M.~Paolini, ``Crucial economics for mobile data backhaul,'' in \emph{Whitepaper
  available at
  http://cbnl.com/sites/all/files/userfiles/files/CB-002070-DC-LATEST.pdf},
  2011.

\bibitem{5473878}
O.~Tipmongkolsilp, S.~Zaghloul, and A.~Jukan, ``The evolution of cellular
  backhaul technologies: Current issues and future trends,'' \emph{IEEE
  Communications Surveys \& Tutorials}, vol.~13, no.~1, pp. 97--113, First
  Quarter 2011.

\bibitem{6777766}
Y.~Li, M.~Pioro, and V.~Angelakisi, ``Design of cellular backhaul topology
  using the fso technology,'' in \emph{Proc. of 2nd International Workshop on
  Optical Wireless Communications (IWOW' 2013), Newcastle Upon Tyne, UK}, Oct
  2013, pp. 6--10.

\bibitem{5185525}
S.~Chia, M.~Gasparroni, and P.~Brick, ``The next challenge for cellular
  networks: {B}ackhaul,'' \emph{IEEE Microwave Magazine}, vol.~10, no.~5, pp.
  54--66, August 2009.

\bibitem{6226966}
S.~Korotky, ``Price-points for components of multi-core fiber communication
  systems in backbone optical networks,'' \emph{IEEE/OSA Journal of Optical
  Communications and Networking}, vol.~4, no.~5, pp. 426--435, May 2012.

\bibitem{5771213}
F.~Demers, H.~Yanikomeroglu, and M.~St-Hilaire, ``A survey of opportunities for
  free space optics in next generation cellular networks,'' in \emph{Proc. of
  IEEE 9th Annual Communication Networks and Services Research Conference
  (CNSR' 2011), Ottawa, ON, Canada}, May 2011, pp. 210--216.

\bibitem{1495057}
A.~Kashyap and M.~Shayman, ``Routing and traffic engineering in hybrid {RF/FSO}
  networks,'' in \emph{Proc. of the IEEE International Conference on
  Communications (ICC' 2005), Seoul, Korea}, vol.~5, May 2005, pp. 3427--3433.

\bibitem{65841515}
N.~Letzepis and A.~G.~I. Fàbregas, ``Hybrid rf/fso communications,'' in
  \emph{Advanced Optical Wireless Communication Systems}, S.~Arnon, J.~Barry,
  G.~Karagiannidis, R.~Schober, and M.~Uysal, Eds.\hskip 1em plus 0.5em minus
  0.4em\relax Cambridge University Press, 2012, pp. 273--302.

\bibitem{6876609}
J.~Rak and W.~Molisz, ``Reliable routing and resource allocation scheme for
  hybrid {RF}/{FSO} networks,'' in \emph{Proc. of 16th International Conference
  on Transparent Optical Networks (ICTON' 2014), Graz, Austria}, July 2014, pp.
  1--4.

\bibitem{16546169}
V.~W.~S. Chan, ``Free-space optical communications,'' \emph{IEEE Journal of
  Lightwave Technology}, vol.~24, no.~12, pp. 4750--4762, Dec 2006.

\bibitem{5545666}
M.~Zotkiewicz, W.~Ben-Ameur, and M.~Pioro, ``Finding failure-disjoint paths for
  path diversity protection in communication networks,'' \emph{IEEE
  Communications Letters}, vol.~14, no.~8, pp. 776--778, August 2010.

\bibitem{254515158}
Z.~Ghassemlooy, W.~Popoola, and S.~Rajbhandari, \emph{Optical Wireless
  Communications: System and Channel modelling with Matlab}.\hskip 1em plus
  0.5em minus 0.4em\relax CRC Press, 2012.

\bibitem{6844864}
M.~Khalighi and M.~Uysal, ``Survey on free space optical communication: A
  communication theory perspective,'' \emph{IEEE Communications Surveys \&
  Tutorials}, vol.~16, no.~4, pp. 2231--2258, Fourth quarter 2014.

\bibitem{168714992012}
D.~Borah, A.~Boucouvalas, C.~Davis, S.~Hranilovic, and K.~Yiannopoulos, ``A
  review of communication-oriented optical wireless systems,'' \emph{EURASIP
  Journal on Wireless Communications and Networking}, vol. 2012, no.~1, p.~91,
  2012.

\bibitem{1495122}
J.~Llorca, A.~Desai, and S.~Milner, ``Obscuration minimization in dynamic free
  space optical networks through topology control,'' in \emph{Proc. of Military
  Communications Conference (MILCOM' 2004), Monterey, CA, USA}, vol.~3, Oct
  2004, pp. 1247--1253.

\bibitem{Smadi:09}
M.~N. Smadi, S.~C. Ghosh, A.~A. Farid, T.~D. Todd, and S.~Hranilovic,
  ``Free-space optical gateway placement in hybrid wireless mesh networks,''
  \emph{J. Lightwave Technol.}, vol.~27, no.~14, pp. 2688--2697, Jul 2009.

\bibitem{4609027}
V.~Rajakumar, M.~Smadi, S.~Ghosh, T.~Todd, and S.~Hranilovic, ``Interference
  management in {WLAN} mesh networks using free-space optical links,''
  \emph{IEEE Journal of Lightwave Technology}, vol.~26, no.~13, pp. 1735--1743,
  July 2008.

\bibitem{4357553}
D.~Wang and A.~Abouzeid, ``Throughput capacity of hybrid radio-frequency and
  free-space-optical ({RF/FSO}) multi-hop networks,'' in \emph{Proc. of IEEE
  Information Theory Workshop (ITW' 2007), Bergen, Norway}, Jan 2007, pp.
  3--10.

\bibitem{5462107}
I.~K. Son and S.~Mao, ``Design and optimization of a tiered wireless access
  network,'' in \emph{Proc. of IEEE Conference on Computer Communications
  (INFOCOM' 2010), San Diego, CA, USA}, March 2010, pp. 1--9.

\bibitem{4746591}
X.~Cao, ``An integer linear programming approach for topology design in {OWC}
  networks,'' in \emph{Proc. of IEEE Global Telecommunications Workshops
  (GLOBECOM' 2008), New Orleans, LA, USA}, Nov 2008, pp. 1--5.

\bibitem{6134071}
F.~Ahdi and S.~Subramaniam, ``Optimal placement of {FSO} links in hybrid
  wireless optical networks,'' in \emph{Proc. of IEEE Global Telecommunications
  Conference (GLOBECOM' 2011), Houston, Texas, USA}, Dec 2011, pp. 1--6.

\bibitem{6844494}
Y.~Li, N.~Pappas, V.~Angelakis, M.~Pi{\'{o}}ro, and D.~Yuan, ``Optimization of
  free space optical wireless network for cellular backhauling,'' \emph{Arxiv
  e-prints}, vol. abs/1406.2480, 2014.

\bibitem{2514014}
Y.~Li, N.~Pappas, V.~Angelakis, M.~Pioro, and D.~Yuan, ``Resilient topology
  design for free space optical cellular backhaul networking,'' in \emph{Proc.
  of IEEE Global Telecommunications Conference Workshops (GLOBECOM' 2014),
  Austin, Texas, USA}, 2014.

\bibitem{Dahrouj_backnet_magazine}
H.~Dahrouj, A.~Douik, T.~Y. Al-Naffouri, and M.-S. Alouini, ``Cost-effective
  hybrid rf/fso backhaul solution for next generation wireless systems,''
  \emph{accepted in IEEE Wireless Communications Magazine}, Oct. 2015.

\bibitem{Hybrid_Douik_ICC16}
A.~Douik, H.~Dahrouj, T.~Y. Al-Naffouri, and M.-S. Alouini, ``Resilient
  backhaul network design using hybrid radio/free-space optical technology,''
  \emph{submitted to IEEE International Conference on Communication (ICC' 16)},
  Oct. 2015.

\bibitem{9874286}
V.~Vassilevska, ``Efficient algorithms for clique problems,'' \emph{Inf.
  Process. Lett.}, vol. 109, no.~4, pp. 254--257, 2009.

\bibitem{16513519}
P.~R.~J. Ostergard, ``A fast algorithm for the maximum clique problem,''
  \emph{Discrete Appl. Math}, vol. 120, pp. 197--207.

\bibitem{13265492}
K.~Yamaguchi and S.~Masuda, ``A new exact algorithm for the maximum-weight
  clique problem,'' in \emph{Proc. Of the 23rd International Technical
  Conference on Circuits/Systems, Computers and Communications (ITC-CSCC'
  2008), Yamaguchi, Japan}.

\bibitem{6607889}
Z.~Akbari, ``A polynomial-time algorithm for the maximum clique problem,'' in
  \emph{Proc. of IEEE/ACIS 12th International Conference on Computer and
  Information Science (ICIS' 2013), Niigata, Japan}, June 2013, pp. 503--507.

\bibitem{25181258}
B.~Mohar, \emph{The {L}aplacian {S}pectrum of {G}raphs}.\hskip 1em plus 0.5em
  minus 0.4em\relax Wiley, 1991.

\end{thebibliography}

\end{document}